\documentclass{tlp}

\usepackage{xspace}
\usepackage{listings}
\usepackage{url}

\usepackage[T1]{fontenc}

\usepackage[usenames,dvipsnames]{xcolor}
\usepackage{amssymb, amsmath}

\usepackage{graphicx}
\usepackage{tikz} 

\newcommand{\code}[1]{{\tt #1}}

\newcommand{\idp}{{\sc IDP}\xspace}
\newcommand{\idpdrie}{{\sc IDP3}\xspace}
\newcommand{\minisatid}{{\sc MiniSAT(ID)}\xspace}
\newcommand{\gidl}{{\sc GidL}\xspace}
\newcommand{\fodotidp}{{\sc FO($\cdot$)$^{\mathtt{IDP3}}$}\xspace}

\newcommand{\fodot}{{\sc FO($\cdot$)}\xspace}
\newcommand{\fodotid}{{\sc FO(ID)}\xspace}

\newcommand{\bigor}{\ensuremath{\bigor}}

\newcommand{\theory}{\ensuremath{T}\xspace}

\newcommand{\model}{\ensuremath{\mathcal{M}}\xspace}

\usepackage{listings}

\lstdefinelanguage{idp}{
	morekeywords=[1]{namespace,vocabulary,theory,structure,procedure,term},
	morekeywords=[2]{include,using,type,isa,contains,partial,extern,LFD,GFD},
	morekeywords=[3]{int,float,char,string,nat},
	morekeywords=[4]{if,then,else,for,end},
	morecomment=[s]{/*}{*/},	
	morecomment=[l]{//}
}
\usepackage{graphicx}

\newtheorem{definition}{Definition}

\newtheorem{theorem}{Theorem}

\newcommand{\old}[1]{\textcolor{gray}{}}

\newcommand{\hidden}[1]{}

\newcommand{\ignore}[1]{}

\newcommand{\listingref}[1]{Listing~\ref{#1}}

\newcommand{\idpcode}[1]{\lstinline{#1}}

\usepackage{cleveref} 

\title[Predicate Logic as a Modeling Language]
{Predicate Logic as a Modeling Language:\\
Modeling and Solving some Machine Learning and Data Mining Problems with \idpdrie}

\author[Bruynooghe et al.]
{MAURICE BRUYNOOGHE, HENDRIK BLOCKEEL, BART BOGAERTS, \authorbreak
  BROES DE CAT,  STEF DE POOTER, JOACHIM JANSEN,  \authorbreak ANTHONY
  LABARRE, JAN RAMON, MARC DENECKER, \\ Department of
  Computer Science, KU Leuven \and SICCO VERWER \\Radboud Universiteit
  Nijmegen, Institute for Computing and Information Sciences \\
\email{firstname.secondname@cs.kuleuven.be, siccoverwer@gmail.com}}

 \submitted{March 22, 2013}
 \revised{Februari 3, 2014}
 \accepted{March 6, 2014}

\begin{document}
\maketitle
\begin{abstract}

  This paper provides a gentle introduction to problem solving with
  the \idpdrie system. The core of \idpdrie is a finite model
  generator that supports first order logic enriched with types,
  inductive definitions, aggregates and partial functions. It offers
  its users a modeling language that is a slight extension of
  predicate logic and allows them to solve a wide range of search
  problems.
  Apart from a small introductory example, applications are selected
  from problems that arose within machine learning and data mining
  research. These research areas have recently shown a strong interest
  in declarative modeling and constraint solving as opposed to
  algorithmic approaches. The paper illustrates that the \idpdrie
  system can be a valuable tool for researchers with such an
  interest.

  The first problem is in the domain of stemmatology, a domain of
  philology concerned with the relationship between surviving variant
  versions of text.
  The second problem is about a somewhat related problem within
  biology where phylogenetic trees are used to represent the evolution
  of species.
  The third and final problem concerns the classical problem of learning
  a minimal automaton consistent with a given set of strings. For this
  last problem, we show that the performance of our solution comes
  very close to that of a state-of-the art solution. 
  For each of these applications, we analyze the problem, illustrate
  the development of a logic-based model and explore how alternatives
  can affect the performance.
\end{abstract}


\section{Introduction}


In his seminal paper, Kowalski~\citeyear{ifip/Kowalski74} proposed to
use first order predicate logic (FO) as a programming language. He
argued that it is possible to use deduction for computation by
associating a procedural interpretation to the Horn clause subset of
first-order logic.  These ideas found their incarnation in the
language Prolog.

Whereas Prolog uses deduction as inference method, other
inference methods exist. Most prominent is model generation as used in
propositional SAT solvers. Also the inference method of Constraint
Programming can be considered as model generation; indeed, its solvers
attempt to assign values to variables while satisfying a set of
constraints.
The last decades have witnessed tremendous progress in solver
technology for Constraint Programming (CP) and SAT solving. In CP,
this progress is at the basis of a shift from Constraint Programming
to Constraint Modeling\footnote{In this paper, we use the word model
  in two different meanings. Firstly, a
  model is  a structure that satisfies the theory, as in ``model generation''.  Secondly, a model is the result of modeling a problem domain. It is  a theory in logic, a formal specification of the problem domain. It should be clear from the context what is   intended.}.
Notorious examples are {\sc
  Essence}~\cite{constraints/FrischHJHM08} and
Zinc~\cite{constraints/MarriottNRSBW08}.  Within logic programming,
the introduction of stable semantics~\cite{iclp/GelfondL88} eventually
led to the Answer Set Programming (ASP)
paradigm~\cite{cacm/BrewkaET11} that, similar to SAT, uses model
generation instead of deduction for inference. Many ASP-based systems
exist. Examples are DLV~\cite{tocl/LeonePFEGPS02},
\texttt{clasp}~\cite{lpnmr/GebserKNS07a} and Smodels~\cite{lpnmr/SyrjanenN01}.

All this progress raises the question what is the status of logic as a
modeling language. SAT is restricted to propositional logic. It can be
considered as the assembler language for modeling. Indeed, there
are many examples of programs that generate SAT encodings to obtain
state-of-the-art solvers for various classes of problems. One can find
examples in the areas of planning and of generating deterministic
finite automata, to name just a few. However, SAT is not suited as a
language for developing models.  For what concerns ASP, it is an
expressive high level language but it is not based on predicate logic. Today,
many intricacies of stable model semantics are hidden in high level
ASP constructs such as constraints and choice rules; however, its two
forms of negation (``not'' and strong negation)~\cite{cacm/BrewkaET11}
clearly distinguish it from first-order logic; the deviation from
first-order-logic semantics could be an obstacle for newcomers.

Historically, predicate logic was always viewed as a very expressive
modeling language. This is remarkable given that anyone who used it
for modeling a practical domain will have experienced its
inconvenience in expressing certain common propositions. A clear
weakness is in expressing inductively definable concepts such as the
transitive closure of a binary relation. Another deficiency is in
expressing bounds on the cardinality or the sum of sets. Practical
modeling languages in Constraint Programming or ASP therefore support
some of these propositions. While ASP can express inductive
definitions, it is built on radically different foundations than FO. A
more conservative solution that preserves FO's foundations is to
extend it with suitable language constructs. For instance, it was
argued in several works, for example by~\citeN{tocl/DeneckerT08} and
\citeN{DeneckerV/KRR14}, that a rule set formalism under an extension
of the well-founded semantics~\cite{GelderRS91} is a natural formalism
to express the most common forms of inductive definitions. Such a
formalism can be integrated with FO in a conceptually clean way.  The
resulting logic was named \fodotid by~\citeN{tocl/DeneckerT08}. The
link between \fodotid and ASP was recently studied by
\citeN{Denecker12}.  Below, we use the notation \fodot to denote the
family of extensions of first-order logic.

In this paper, we explore the use of \fodotidp, the instance of the
\fodot family that is supported by \idpdrie, the current version of
the \idp Knowledge Base System~\cite{inap/DePooterWD11}. \fodotidp
extends first-order logic with inductive definitions, partial
functions, types and aggregates. The \idp system supports model
generation and model expansion \cite{MitchellT05,atcl/Wittocx13}  as inference
methods and is one of the fastest such systems \cite{LPNRM/Calimeri11}. Particular to the modeling language used here is the
combination of a purely declarative modeling language with a
procedural language that handles the interaction with the outside
world. Indeed, in contrast to Prolog, the control of the search can be
left to the solver and the user can concentrate on modeling.  As we
will illustrate in the paper, that does not mean that any correct
model will do; when performance matters, models have to be designed
with care. 

As for the organization of the paper, we start Section~\ref{sec:fodot}
with recalling the \fodot family of extensions of predicate
logic. Next, we introduce the \fodotidp instance of \fodot and the
\idp knowledge base system that supports \fodotidp as a modeling
language. The section continues with the shortest path problem as an
illustrative example.  After describing a very basic model, it
explores how alternative models affect the performance of the
underlying solver.

The other sections explore the use of the IDP system for solving some
real-world problems encountered in the domain of machine learning and
data mining by some of the authors.  Researchers in this domain have
become increasingly aware of the fact that data analysis problems come
in many different variants, which do not always fit the standard
algorithms well.  It is infeasible to develop algorithms for each
specific variant, but recently it has been shown that some standard
data mining problems, as well as their variants, can be modeled as
constraint problems, and solved by general-purpose solvers with
performance comparable to that of dedicated algorithms
\cite{DBLP:journals/ai/GunsNR11}.  Our discussion on the use of IDP for three
different tasks adds support for the claim that declarative modeling
may have an important role to play in machine learning and data
mining.

The first task, addressed in Section~\ref{sec:stemma} is in the domain
of a stemmatology, a part of philology that studies the relationship
between surviving variant versions of a text. A stemma is a family
tree that shows how different copies of the same text relate to each
other. These copies ---manuscripts--- are not identical, but
evolve. Manuscripts often do not have a single parent, different parts
can be copied from different parents, so a stemma is in fact a
directed acyclic graph. A typical task is to analyze the plausibility
of a stemma hypothesis. For this task, the philologist collects
datasets describing features of the text. The values of a feature
represent variant readings of a fragment of the text. A common
assumption is that the stemma has, for each variant, a unique
manuscript that is the source of the variant. The feature value is
unknown for some manuscripts, and the question is whether these
unknown values can be assigned such that there is indeed a unique
source for each feature value. In working out this task, we also
illustrate how the procedural side of \fodotidp allows the user to
organize a complete workflow.

The second task (Section~\ref{sec:CS}), although in the very different
domain of biology, is somewhat related to the previous one as it is
concerned with phylogenetic trees. Phylogeny is an area in which many
problems arise. Several problems have been tackled by means of Answer
Set Programming, see \citeN{Erdem11} for an overview. Here we address a
new problem in this area. Phylogenetic trees have in
their leaves a set of current species and the tree represents the
evolutionary relationship between them. Often there are several
equally plausible evolutionary explanations and hence different
phylogenetic trees. The question addressed here is: what is the
minimal supergraph that represents each of the individual trees. What
makes the problem difficult is that the correspondence between the
internal nodes of the different trees is unknown and has to be
guessed. Different guesses result in different supergraphs.

In Section~\ref{sec:dfa}, we study the well-known problem of learning
a minimal deterministic finite state automaton (DFA) that is
consistent with a given set of accepted and rejected strings. This is
a classical machine learning task for which competitions are
organized. A state-of-the-art method~\cite{Verwer,Verwer12}, winner of
the 2010 Stamina competition \cite{stamina}, solves it by a
problem-specific program that iteratively creates a SAT encoding and
applies a SAT-solver for an increasing number of states until a model
is found.
Here we explore to what extent a high  level \fodot formalization
can compete with a laboriously constructed encoding as a propositional
SAT problem.


These three problems can be abstracted as graph problems and are
NP-complete. Solving them inherently involves search; heuristics
are needed to guide the search towards solutions. Developing an
algorithm in a procedural language is time consuming, error-prone and
challenging. The use of a declarative modeling language liberates the
programmer from the task and allows him to devote more time to the
proper formalization. Moreover, the default heuristics of the
underlying solvers are often sufficient to obtain adequate solutions.

In Section~\ref{sec:concl}, we reflect on our achievements and discuss
where there is potential for further improvement.  Some of the
material in this paper is based on work of \citeN{iclp/Blockeeletall12} and,
for stemmatology, on  work of \citeN{Andrewsetal12}.

\section{\fodot and the \idp system}\label{sec:fodot}

First-order logic has a long tradition and a well understood semantics
but also some limitations with regard to its expressiveness, which
makes it not so well suited as a language for knowledge
representation. The most notorious problem is that it cannot naturally express transitive closures such as ``$x$ is reachable from $y$ if either $x$
and $y$ are connected or there exists a $z$ such that $x$ and $z$ are
connected and $x$ is reachable from $z$''.
Note that Prolog programmers can cope with transitive closure and that
the least Herbrand interpretation captures its meaning; actually, in
the first years of logic programming, many Prolog programmers did not
realize it was an issue in the knowledge representation community; at
the same time, many in the latter community were ignorant about
Prolog's expressiveness.
In the knowledge
representation community, there are two ways to work around the
limitation. On the one hand, one can introduce knowledge
representation languages with a semantics different from first-order
logic; on the other hand, one can enhance first-order logic with
additional constructs.  The former approach is taken by the ASP
community; the latter approach has been advocated in
\cite{tocl/DeneckerT08} where first-order logic was extended with (not
necessarily monotone) inductive definitions. It was argued that this
extension resulted in a very natural and expressive language whose
meaning was captured by a generalization of the well-founded semantics
introduced by \citeN{GelderRS91}. This extension was named \fodotid
and later work used the notation \fodot for a family of languages
extending first-order logic.

The FO extension used as modeling language throughout this paper
includes not only inductive definitions, but also partial functions,
types and aggregates. It is the extension supported by the \idpdrie
version of the \idp Knowledge Base System that was for the first time
described by \citeN{inap/DePooterWD11}\footnote{The examples used
  throughout the paper make use of IDP version 3.2.0.}. We denote this
extension as \fodotidp.

Functions, of which constants are a special case, are very convenient
in modeling. They are used in all the modeling examples of the
paper. While $n$-ary functions can be considered as syntactic sugar
for predicates with $n+1$ arguments, the use of functions makes models
more concise and readable. Indeed, when functions are represented as
predicates, the functional dependency between the input arguments and
the result needs to be represented as a separate constraint. Partial
functions give extra flexibility to the modeler. An example can be
found in the model of finite state automata in Section~\ref{sec:dfa}
where a state doesn't need a transition for all symbols in the
automaton's input alphabet.

In almost all applications, the universe is not uniform but contains
different types. Relations are typed. Quantification is ``naturally''
typed, namely, we quantify over objects of a type.
It is not difficult to make typing explicit in untyped predicate logic.
Yet, it requires an extra discipline of the user to make the
types in her quantifications and of her relations and functions explicit.
By introducing an explicit type system, even a simple many-sorted type
system, and a type checking and inference system (to discover type
clashes and to guess the types of variables and/or parameters), theories
become more compact and graceful. Moreover, a number of bugs can be
detected: syntactic errors in variable names, swapped or missing
arguments, unintended reuse of variables, etc. This is common wisdom. Indeed,
well-typed theories go wrong less often \cite{journals/jcss/Milner78}. The type system of
IDP3 is not needed from a computational point of view, but for above-mentioned
reasons, we often find it convenient.

Aggregates are another extension that contribute to the readability
and conciseness of models. Consider for example the constraint
expressing the functional dependency. One needs to express that there
is exactly one value (or in the case of a partial function at most one
value) for each combination of input arguments. The availability of
aggregates makes it a lot more convenient to express such
constraints. A study about the semantics of aggregates in definitions,
including the case of recursion, has been made by 
\citeN{tplp/PelovDB07}.

\subsection{The logical components of an \fodotidp model}

In this section we introduce the basic notions of an \fodotidp
model. We restrict ourselves to what is needed to understand the
examples later on in the paper.

An \fodotidp model consists of a number of logical components, namely
vocabularies, structures, terms, and theories. A {\em vocabulary}
declares the symbols to be used\footnote{Contrary to Prolog and ASP,
  the first character of a symbol has no bearing on its kind.}. A {\em
  structure} is used to specify the domain and the data; it can be viewed as a sort of database,  it provides a partial
(three-valued) interpretation of the symbols in the vocabulary. 
In the context of optimization problems, a {\em term} component declares
the numerical cost term to be optimized.
 A {\em theory} consists of FO formulas and definitions.
A {\em  definition} is a set of {\em rules} of the form
$\forall \bar x: P(\bar x) \leftarrow \varphi(\bar x).$
where $\varphi$ is an \fodotidp formula\footnote{Definitions have a
  lot in common with pure Prolog rules.}.
An \fodotidp \emph{formula} differs from FO formulas in two ways.
Firstly, \fodotidp is a many-sorted logic: every variable has an
associated {\em type} and every type an associated domain.  Moreover,
it is order-sorted: types can be subtypes of others.  Secondly,
besides the standard terms in FO, \fodotidp formulas can also have
aggregate terms: functions over a set of domain elements and
associated numeric values which map to the sum, product, cardinality,
maximum or minimum value of the set.

We write $\model\models\theory$ to denote that structure \model
satisfies theory \theory. With $x^\model$, we denote the
interpretation of $x$ under $\model$, where $x$ can be a formula or a
term.
Without going in full details, \model satisfies \theory when (i) every
FO formula $F$ of \theory is satisfied in \model ($F^\model$ is true)
and (ii) every definition of \theory is satisfied in \model. A
structure \model satisfies a definition $D$ when the well-founded
model construction on $D$ \cite{GelderRS91} that starts from $O$, the
restriction of $M$ to the predicates not defined in $D$, results in
\model. See \citeN{warrenbook14} for more details.

\subsection{The \idpdrie system}\label{sec:framework}

The \idpdrie system~\cite{inap/DePooterWD11} is a \emph{Knowledge Base
  System} (KBS) that intends to offer the user a range of inference
methods such as {model expansion}, {optimization}, {verification},
{symmetry breaking} and {grounding} and to make use of different state
of the art technologies including SAT, SAT Modulo
Theories~\cite{jacm/NieuwenhuisOT06}, Constraint Programming and
various technologies from Logic Programming.

In this paper, we make use of the inference methods {model expansion},
{satisfiability checking} and {model minimization}.  The most
important inference method is \emph{model expansion} discussed by
\citeN{MitchellT05} and further extended by \citeN{atcl/Wittocx13}.  The
idea of model expansion is to extend a partial structure (an
interpretation) into a full structure that satisfies all constraints
specified by the \fodotidp model. More formally, the task of model
expansion is, given a vocabulary $V$, a theory $T$ over $V$ and a
partial structure $S$ over $V$ (at least interpreting all types), to
find a structure \model that satisfies $T$ and expands $S$, i.e.,
\model is a model of the theory and the input structure $S$ is a
subset of \model.  In the \idpdrie system, this task is executed by
\idpcode{modelexpand(T,S)}. The result of the modelexpand procedure is
a list of models of $T$ that expand $S$. If the option
\idpcode{nbmodels} is set to a value $n$ different from $0$, \idpdrie will
stop searching for more models once it has found $n$
models.

{\em Satisfiability checking} is related to model
expansion. Calling \idpcode{sat(T,S)} in the \idpdrie system will
return true if and only if \idpcode{modelexpand(T,S)} would have
returned at least one model. However, since we are not interested in
the actual models, some optimizations can be done to speed up this
inference.

In case of {\em model minimization}, also a numerical cost term $t$ is
given. The task is to find a model \model of $T$ that expands $S$ such
that, for all other models \model' expanding $S$, $t^{\model} \leq
t^{\model'}$. Model minimization is activated by
\idpcode{minimize(T,S,t)} with $t$ referring to the term component
defining the term.

The \idpdrie system allows users to specify \fodotidp problem
descriptions. The basic overall structure of the logical components is
as in the following schema.
\begin{lstlisting}
	vocabulary V   { ... }  	theory T: V    { ... }
	term t: V      { ... }  	structure S: V { ... }
\end{lstlisting}
This schema defines a vocabulary \idpcode{V} which is then used as a
context in the theory \idpcode{T}, the term \idpcode{t} and the
structure \idpcode{S}. In general, several vocabularies can be
defined, even extending other vocabularies.

We use IDP syntax in the examples throughout the paper. Each IDP
operator has an associated logical operator, the main (non-obvious)
operators being: \code{\&}($\land$), \code{|}($\lor$),
$\sim$($\lnot$), \code{!}($\forall$), \code{?}($\exists$),
\code{<=>}($\equiv$), $\sim$\code{=}($\neq$).

A distinguishing feature of \fodotidp models is that they not only
consist of logical components but also have one or more procedural
components. These procedural components consist of procedural code
that can perform actions. Actions include the execution of an
inference method on a particular logical theory, but also the
presentation of results to the user. Procedures allow to glue together
a sequence of actions in a process that performs a task for the
user. The convention is that the user's task is performed by invoking
the procedure \idpcode{main()}. Such a task can start with procedural
code to prepare one or more structures from input files or databases,
can continue with performing a number of inference task on
combinations of theories with structures and can end with presenting
the results to the user. The procedural language has to be a flexible
and extensible scripting language that offers a smooth integration
with the C++ solvers of the \idp system. The IDP
system~\cite{inap/DePooterWD11} makes use of the
Lua~\cite{SPE/IerusalimschyFC96} scripting language for this
purpose. It allows us to treat the various logical components of an
\fodotidp theory as objects that can be manipulated from within the
procedures.

More information on the \idp system and in particular its \idpdrie
version is given by \citeN{lash08/WittocxMD08} and
\citeN{inap/DePooterWD11} and can be found at
\url{http://dtai.cs.kuleuven.be/krr/software/idp3}.

\subsection{An example: The shortest path problem}\label{sec:shrtpth}

As an illustration, we model the shortest path problem
(Listing~\ref{shortestpath}).  The {\em vocabulary} consists of a
single type, two constants and three predicates. The {\em structure}
specifies the given graph: the interpretation of the type
\texttt{node} (the domain elements \texttt{A, B, C,} and \texttt{D})
and of the predicate \texttt{edge(node,node)} (the domain atoms
\texttt{edge(A,B), edge(B,C), edge(C,D), and \texttt{edge(A,D)})}, as
well as the constants \texttt{from} (the domain element \texttt{A})
and \texttt{to} (the domain element \texttt{D}) which identify the
begin- and endpoint of the path searched for. The predicate
\texttt{edgeOnPath(node,node)} is used to represent the edges that
participate in a shortest path. It provides the base case of the
transitive relation \texttt{reaches(node,node)} which is defined in
the {\em theory} component. Definitions are given between ``\{'' and
``\}''.  Note that we use the most basic definition for the transitive
closure: we join the \texttt{reaches} relation with itself.

Besides this inductive (recursive) definition, the theory also
specifies the constraints expressing that the \texttt{edgeOnPath/2}
atoms included in a model of the theory indeed compose a simple path
from \texttt{from} to \texttt{to}. The first constraint, a universally
quantified implication, ensures that \texttt{edgeOnPath/2} atoms are
indeed \texttt{edge/2} atoms. This constraint explicitly mentions the
type of the quantified variables; however, this is optional; these
types can be inferred from the type declarations of the predicates of
the formula. The types of the quantified variables are omitted in the
following constraints.  The second constraint, a simple fact, imposes
that \texttt{reaches/2} includes the pair \texttt{(from,to)}. The
third constraint, a conjunction of negated formulas, states that the
\texttt{edgeOnPath/2} atoms should neither include an edge arriving in
\texttt{from} nor an edge leaving \texttt{to}.  The fourth constraint
is a universally quantified conjunction of two cardinality
constraints; it expresses that every node has less than 2 incoming
edges and less than 2 outgoing edges (i.e., that the path is
simple). The notation {\tt ?$<$2 y : edgeOnPath(y,x)} means that there
are strictly less than 2 y's that have an edge to x in the path. This
is syntactic sugar for the aggregate {\tt \#\{ y : edgeOnPath(y,x)\} <
  2}. This aggregate is a more concise formulation of the FO
constraint {\tt ! y1 y2 : edgeOnPath(y1,x) \& edgeOnPath(y2,x) =>
  y1=y2}.  Finally, the last constraint, another universally
quantified implication, states that the endpoints of selected edges
are reachable from \texttt{from} (i.e., that no edges are selected
that do not contribute to the path).

The {\em term} component of the model defines the term
\texttt{lengthOfPath} as an aggregate expression counting the number
of tuples in the \texttt{edgeOnPath} relation of a model of the
theory. Minimizing this term ensures that the path described in a
model of the theory is indeed a shortest path.

The {\em procedure} component
shows Lua code invoking the model minimization task and printing the
result. The Lua code treats the logical components as first-class
citizens and uses them as parameters in the method call activating the
solver. The annotation \texttt{[1]} directs the solver to return at
most one solution.


\begin{lstlisting}[caption={Calling main() solves the shortest path problem for the given data.},label={shortestpath}]
vocabulary sp_voc {
    type node
    from,to: node
    edge(node,node)
    edgeOnPath(node,node)
    reaches(node,node)
}
theory sp_theory1 : sp_voc {
    {! x y : reaches(x,y) <- edgeOnPath(x,y).
     ! x y z : reaches(x,y) <- reaches(x,z) & reaches(z,y).}

    ! x[node] y[node] : edgeOnPath(x,y) => edge(x,y). // (1)
    reaches(from,to). // (2)
    ~(? x : edgeOnPath(x,from)) & ~(? x : edgeOnPath(to,x)).//(3)
    ! x : (?<2 y : edgeOnPath(y,x)) &
          (?<2 y : edgeOnPath(x,y)). // (4)
    ! x y : edgeOnPath(x,y) => reaches(from,y). // (5)
}
structure sp_struct : sp_voc {
    node = {A..D} // shorthand for A,B,C,D
    edge = {A,B; B,C; C,D; A,D} // `;' separated list of tuples
    from = A
    to = D
}
term lengthOfPath : sp_voc { #{ x y : edgeOnPath(x,y) } }
procedure main() {
    /* Search a minimal model */
    sol = minimize(sp_theory1,sp_struct,lengthOfPath)[1]
    /* If no result is returned, no models exist */
    if (sol == nil)
    then print("No models exist.\n")
    else print(sol)
    end
}
\end{lstlisting}

The \idpdrie system performs model expansion and model minimization by
first reducing the problem to extended CNF, using the grounder
\gidl~\cite{jair/WittocxMD10} and subsequently calling the solver
\minisatid~\cite{sat/MarienWDB08}. The grounding process can be
fine-tuned using options for symmetry
breaking \cite{ictai/DevriendtBMDD12}, grounding with bounds
\cite{jair/WittocxMD10}, lazy grounding \cite{iclp/DeCatDS12},
etc. The solver is an extension of MiniSat \cite{sat/EenS03} with
support for aggregate expressions, inductive definitions and
branch-and-bound optimization. Recently, \minisatid was extended to
offer support for finite domain constraints, using the propagation
techniques described in \cite{toplas/SchulteS08} or, alternatively,
interfacing with the {\sc Gecode} Constraint Programming engine.

\subsection{Exploring the space of  models for the shortest path problem}\label{sec:user}

The above is a correct \idpdrie model that provides a declarative
solution to the problem at hand.
However, if performance matters, or the instances
are that large that the grounding cannot fit in memory, then other
models can be preferable. In the long run, perhaps an optimizer can
transform a simple model in a model for which model generation has a
better performance, but, with the current state of affairs, it is up
to the user to explore the design space and to look for better
performing models. We do so in this section for the shortest path
problem.

\begin{figure}[t]
  \centering
\includegraphics[width=\textwidth]{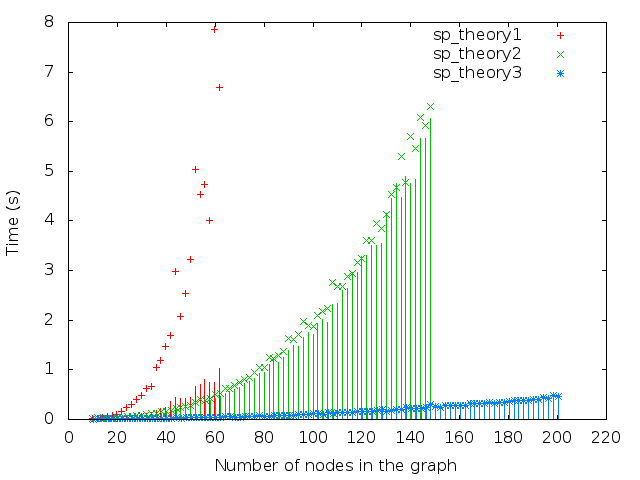}

  \caption{Different grounding and total running times for
    \idpcode{sp\_theory1}, \idpcode{sp\_theory2} and
    \idpcode{sp\_theory3}. Experiments are performed on graphs with an
    increasing number of nodes but a constant edge density. The
    vertical bar shows the grounding time, the part above the vertical
    bar is the solving time.}
  \label{fig:experiments1}
\end{figure}

\begin{lstlisting}[caption={Another definition for {\tt reaches}.},label={betterDef}]
theory sp_theory2 : sp_voc {
  { reaches(x,y) <- edgeOnPath(x,y).
    reaches(x,y) <- edgeOnPath(x,z) & reaches(z,y). }
  /* ... constraints as in sp_theory1 ... */
}
\end{lstlisting}

Prolog programmers would never define the {\tt reaches/2} predicate as
in theory {\tt sp\_theory1} of \listingref{shortestpath} but rather as
in theory {\tt sp\_theory2} of \listingref{betterDef}. Indeed, apart
from the risk of entering an infinite loop, {\tt reaches(x,z)} can
have many more solutions than {\tt edgeOnPath(x,z)} (which is bounded
by the number of {\tt edges}) and hence, the search space for Prolog's
proof procedure can be substantially larger\footnote{When using a
  system with tabling, the order in the body of the recursive clause
  is better inverted \cite{tplp/SwiftW12}.}. While model generation
cannot loop in the \idp system, the search space argument remains
valid. Comparing the runtime of both systems (see Figure
\ref{fig:experiments1}) reveals that the heuristics of the underlying
SAT solver do not compensate for the larger search space and that the
version of \listingref{betterDef} is substantially
faster\footnote{Using an Intel$^R$ Core$^{TM}$ i5-3550 CPU at 3.30GHz
  with 7.8 GB RAM running Ubuntu with the \idpdrie default
  options.}. However, there is another factor contributing for the
difference between both versions. Figure~\ref{fig:experiments1} splits
runtime in grounding time (the vertical bar) and solving time (above
the bar). For larger problems, one can observe that the grounding also
takes more time. The explanation is the difference in grounding size
(see Figure~\ref{fig:experimentsgrounding}). The size difference must
be attributed to the grounding of the recursive {\tt reaches(x,z)}
rule.  With $n$ nodes, there are $n^3$ instances of the recursive {\tt
  reaches/2} rule in {\tt sp\_theory1}. As for {\tt sp\_theory2}, it
follows from constraint (1) that {\tt edgeOnPath(x,z)} is false
whenever {\tt edges(x,z)} is false, hence, with $e$ the
number of edges, the number of instances is
limited to $e * n$, which is typically substantially smaller than
$n^3$.

\begin{figure}[t]
  \centering
  \includegraphics[width=\textwidth]{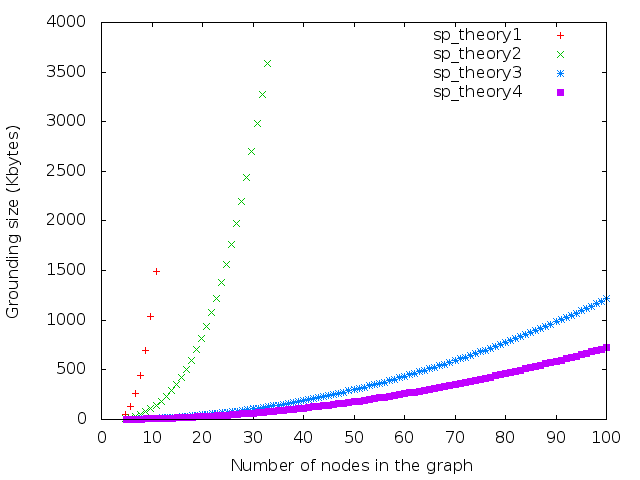}
  \caption{Grounding size for \idpcode{sp\_theory1},
    \idpcode{sp\_theory2}, \idpcode{sp\_theory3} and
    \idpcode{sp\_theory4}. Experiments are performed on graphs with an
    increasing number of nodes but a constant edge density.}
  \label{fig:experimentsgrounding}
\end{figure}

Although the performance has improved quite a bit, we see that it is
still rapidly increasing with the size of the graph. Moreover (see
Figure~\ref{fig:experiments1}) most of the runtime is grounding
time. Can we do better? The recursive {\tt reaches/2} rule has three
variables and remains expensive. The grounding has $n^2$ atoms {\tt
  reaches(n1,n2)}, expressing whether $n1$ and $n2$ are connected
while we are only interested in paths going from {\tt from} to {\tt
  to}. Hence we should better use a unary {\tt reachable} predicate
and define which points are reachable from {\tt from}. The new
versions of the {\tt vocabulary} and the {\tt theory} are shown in
\listingref{decreaseArity}. The {\tt term} and {\tt procedure} parts
are as in \listingref{shortestpath}.

\begin{lstlisting}[caption={A unary {\tt reachable} relation instead
    of the binary {\tt reaches relation}.},label={decreaseArity}]
vocabulary sp_voc2 {
    type node
    from,to: node
    edge(node,node)
    edgeOnPath(node,node)
    reachable(node)
}
theory sp_theory3: sp_voc2 {
    { reachable(from).
      reachable(y) <- edgeOnPath(x,y) & reachable(x). }

    ! x[node] y[node] : edgeOnPath(x,y) => edge(x,y). //(1)
    reachable(to). //(2)
    ~(? x : edgeOnPath(x,from)) & ~(? x : edgeOnPath(to,x)).//(3)
    ! x : (?<2 y : edgeOnPath(y,x)) &
          (?<2 y : edgeOnPath(x,y)). //(4)
    ! x y : edgeOnPath(x,y) => reachable(y). //(5)
}
\end{lstlisting}

As Figure~\ref{fig:experiments1} shows, this modification results in a
dramatic speed-up. Also the grounding size
(Figure~\ref{fig:experimentsgrounding}) is substantially reduced.

Constraints (3) and (4) are cardinality constraints on the number of
edges connected to the same node that can participate in a path. They
are redundant with respect to the minimization of the {\tt
  lengthOfPath} term. Indeed, paths of minimal length satisfy both
constraints. Dropping them, as in \listingref{removeConstraints},
spoils the clarity of the model; however, it further reduces the size
of the grounding as can be seen in
Figure~\ref{fig:experimentsgrounding}. The effect on the runtime is
negligible for small graphs and rather negative for larger ones as one
can observe in Figure~\ref{fig:experiments2}. The explanation is that
the removal of those constraints increases the search space. Indeed, partial
solutions violating constraints (3) and (4) are not immediately
rejected. This causes a lot more variance in the solving times.

\begin{lstlisting}[caption={Removing redundant constraints.},label={removeConstraints}]
theory sp_theory4: sp_voc {
    { reachable(from).
      reachable(y) <- reachable(x) & edgeOnPath(x,y). }

    ! x[node] y[node] : edgeOnPath(x,y) => edge(x,y). // (1)
    reachable(to).
}
\end{lstlisting}

\begin{figure}[t]
 \centering
\includegraphics[width=\textwidth]{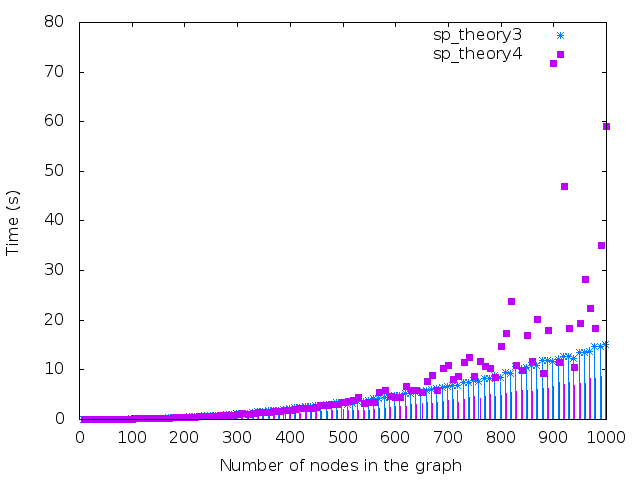}

\caption{Different grounding and total running times for \idpcode{sp\_theory3} and \idpcode{sp\_theory4}. Experiments are performed on graphs with an increasing number of nodes but a constant edge density.}
 \label{fig:experiments2}
\end{figure}

The above example, in which we solved a classical problem with
\idpdrie, illustrates the basic features of \fodotidp. It shows that
problem solving with \idpdrie is a quite different endeavor from
problem solving in other languages. This holds not only for procedural
languages, but also for a declarative language such as Prolog. It also
shows the importance of exploring various models when performance and
memory use matters.

\section{Stemmatology}\label{sec:stemma}

Before the invention of the printing press, texts were copied manually
by scribes.  This copying process was not perfect; scribes often
modified texts, either accidentally or intentionally.  As a result the
surviving copies of many old texts vary significantly.  No text
written before the invention of the printing press, and even up to the
end of the 18th century, when the habit of circulating texts in
manuscript form practically disappeared, can be read without a
preliminary critical analysis of its material witnesses.  This is the
purpose of stemmatology.  The Oxford English Dictionary defines the
field as ``the branch of study concerned with analyzing the
relationship of surviving variant versions of a text to each other,
especially so as to reconstruct a lost original.''

A stemma is a kind of ``family tree'' of a {\em tradition}, a set of
related manuscripts.  It indicates which manuscripts have been copied
from which other manuscripts (``parents''), and which manuscript is
the original source.
It may include both extant (currently existing and available) and
non-extant (``lost'') manuscripts.  The stemma is not necessarily a
tree: sometimes a manuscript has been copied partially from one
manuscript, and partially from another, in which case the manuscript
has multiple parents.

More formally, a stemma can be defined as a CRDAG, a Connected
Directed Acyclic Graph with a single Root \cite{Andrews-IP}.  A
dataset contains the manuscripts from one tradition. Each manuscript
is described by a fixed set of features $F_1$, \ldots, $F_n$, each of
which has a nominal domain $Dom(F_i)$ (variant readings of feature
$F_i$). Typically, a feature refers to a particular location or
section in a text, though it can also be the spelling of a particular
word, e.g., the dwelling of ``Van den Vos Reynaerde'' can be spelled
as Malpertuis, Malpertus, or Malpertuus.

The 19th century philologist Karl Lachmann was among the first to
apply a principled method for reconstructing stemmata from sets of
manuscripts \cite{Timpanaro:book}.  Nowadays, a variety of methods
exist.  Many are borrowed from biology, where a similar problem,
reconstruction of phylogenetic trees, is well-studied.  However, these
methods do not always fit the stemmatological context well.  First,
they assume that phylogenies are tree-shaped, while stemmata are
DAGs\footnote{Some methods return phylogenetic networks, but these
  represent uncertainty about the real tree, which is different from
  claiming that the network represents the actual phylogeny.}.  Second,
these trees contain only bifurcations, while stemmata can have
multifurcations.  Third, in most methods,
the trees are such that each extant copy is at a leaf of the tree,
whereas in stemmatology one extant copy may be an ancestor of another
(and hence should be an internal node).  Fourth, stemmatologists often
have additional information, for instance about the time or place of
origin of a manuscript, which ideally should be taken into account.
Research continues to develop new algorithms better suited for the
stemmatological context \cite{Baret06}.

\subsection{The task}

Apart from reconstructing stemmata from data, stemmatologists are also
interested in other types of analyses, which may, for instance, use a
known stemma or a manually-constructed best-guess stemma as an input.
These types of analysis can be very diverse.  The data mining tasks we
address in this section belong to this category.

The problem studied here assumes that a CRDAG representing a stemma of
a tradition is given, as well as feature data about the manuscripts
from the tradition.  More specifically, the data include a feature for
each location where variation is observed in the tradition represented
by the stemma. For each extant manuscript in the tradition, the
feature data describe its variant reading; the variant reading is
unknown for the non-extant ones.  For most features, it seems rather
unlikely that the same variant reading originated multiple times
independently; i.e., it is reasonable to assume there is one ancestor
where the variant reading occurred for the first time (the ``source''
of the variant). Therefore, we say that {\em the feature is
  consistent with the stemma} if it is possible to indicate for each
variant a single manuscript that may have been the origin of that
variant.  Since for some manuscripts the value of the feature is not
known, checking consistency boils down to assigning a variant to each
node in the CRDAG in such a way that, for each variant, the nodes
having that variant form a CRDAG themselves. Note that one can imagine
exceptions to the above, e.g., a new spelling of a word can be
independently introduced in different copies.

\subsection{Consistency checking of a stemma is NP-complete}

We learned about this problem through contacts with researchers in
stemmatology. One of them had developed an algorithm 
(implemented with a program of about 370 lines of Perl
using a graph library as a back end) to solve the basic task, did
several iterations to handle yet uncovered cases and was still
worried about the completeness of their approach (does the algorithm
always find a solution when a solution exists?). The algorithm
attempts not to make wrong decisions by initially assigning several
variant readings to the non-extant manuscripts and, in a second phase,
remove variant readings while preserving consistency. Once understood,
the problem was formalized as a graph problem and shown to be
NP-complete by one of the authors of this paper. In this
formalization, the variant reading of a text is represented as a color
and checking a stemma is a color-connected problem.

\begin{definition} [Color-connected]
  Two nodes $x$ and $y$ in a colored CRDAG are {\em color-connected}
  if a node $z$ exists ($z$ can be one of $x$ and $y$) such that there
  is a directed path from $z$ to $x$, and one from $z$ to $y$, and all
  nodes on these paths (including $z$, $x$, $y$) have the same color.

  Given a partially colored CRDAG, the {\em color-connected problem} is to
  complete the coloring such that every pair of nodes of the same
  color is color-connected.
\end{definition}

\begin{figure}
\scalebox{1.0}{
\includegraphics[width=\textwidth]{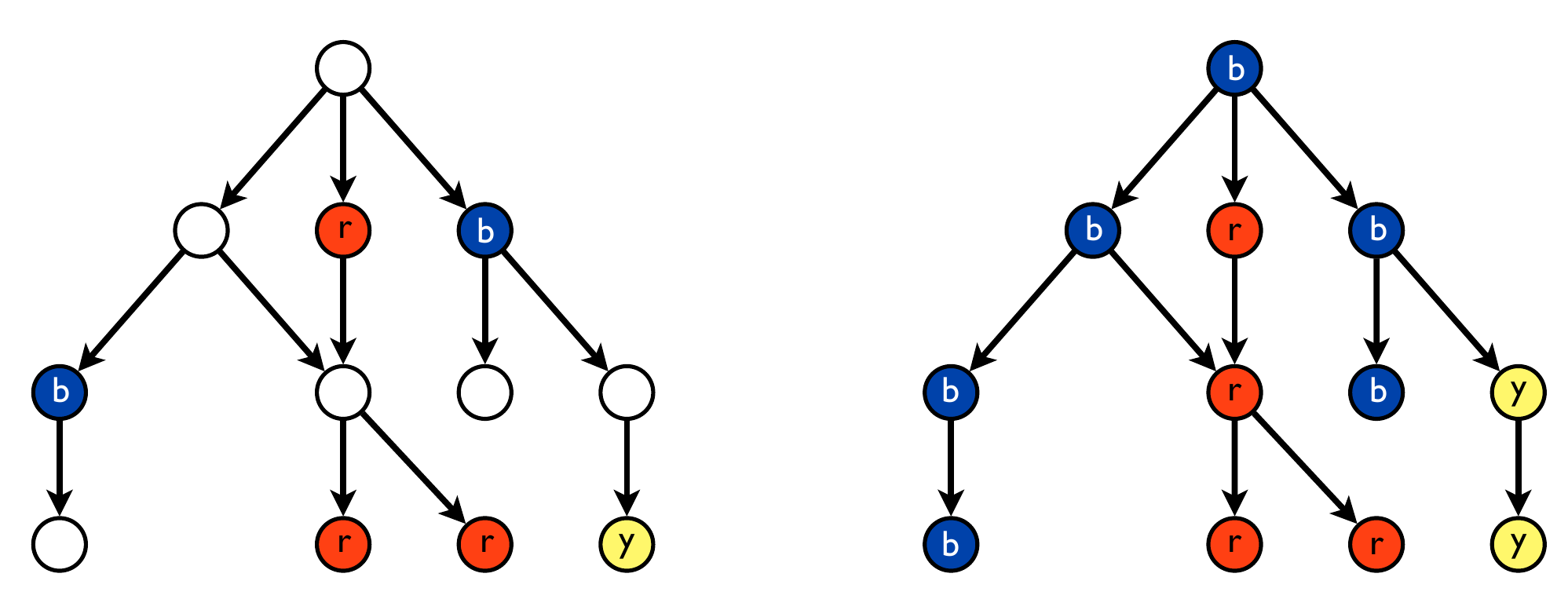}}
\caption{Left: a partial labeling showing for a given feature which
  manuscripts have which variant readings/colors.  Right: a
  complete extension of that labeling where each variant reading is a CRDAG.
  Because such an extension exists, the feature is consistent with
  the stemma.}
\label{fig:colorings}
\end{figure}

An illustration is given in Figure~\ref{fig:colorings}.
A candidate coloring can be checked in polynomial time, hence proving
that the color-connected  problem is NP-hard implies it is NP-complete.

\begin{theorem}
The color-connected  problem is NP-hard.
\end{theorem}

\begin{proof}
  The proof is  by showing a polynomial reduction from SAT to
  color-connectedness.

  There exists a polynomial reduction from a CNF formula to one with
  all clauses either positive (all literals are positive) or negative
  (all literals are negative). Indeed, replace all occurrences of a
  negative literals $\neg x$ by a new positive literals $notx$ and
  add, for every such $notx$ literal the clauses $x \vee notx$ and
  $\neg x \vee \neg notx$.

  So, we assume without loss of generality a CNF formula $T$
  consisting of positive clauses $T^+$ and negative clauses $T^-$. We
  construct a color-connected problem whose solutions correspond to
  the models of $T$.
  Let $T^+ = C^+_1 \wedge C^+_2 \wedge \ldots \wedge C^+_m$ and $T^- =
  C^-_1 \wedge C^-_2 \wedge \ldots \wedge C^-_n$ with $C^+_1 , \ldots,
  C^+_m$ positive clauses and $C^-_1, \ldots, C^+_n$ negative clauses.
  Let $V$ be the set of propositional variables in $T$.

  Now we construct a DAG $G$ consisting of the nodes $V(G) = {r} \cup A \cup B \cup
  V$ where $A=\{a, a_1, a_2, \ldots ,a_m\}$ ($a_i$ stands for clause
  $C^+_i$; $a$ is an extra node) and $B=\{b, b_1, b_2, \ldots
  ,b_{n}\}$ ($b_i$ stands for clause $C^-_i$; $b$ is an  extra
  node). The directed edges are given by $E(G) =\{(v,a_i) | i \in
  [1..m], v \in C^+_i\} \cup \{(v,b_i) | i \in [1..n], v \in C^-_j\}
  \cup \{(a,v)| v \in V\} \cup \{(b,v)| v \in V \cup \{(r,a),(r,b)\} $
  .  Next we color $r$, $a$ and all nodes $a_i$ black and $b$ and all
  nodes $b_i$ white. We obtain 
  a partially colored CRDAG (See Figure~\ref{example-CRDAG} for an
  example). Moreover, a solution to the color-connected problem
  encodes a solution to the original SAT problem. Indeed, each $a_i$
  node, representing a positive clause, is connected with at least
  one black variable. Hence, making all black variables true satisfies
  all positive clauses. Also, each $b_i$ node, representing a
  negative clause, is connected with at least one white
  variable. Hence making the white variables false satisfies all
  negative clauses. It follows that the color-connected problem is
  NP-hard.
\end{proof}

\begin{figure}
\scalebox{0.8}{
\includegraphics[width=\textwidth]{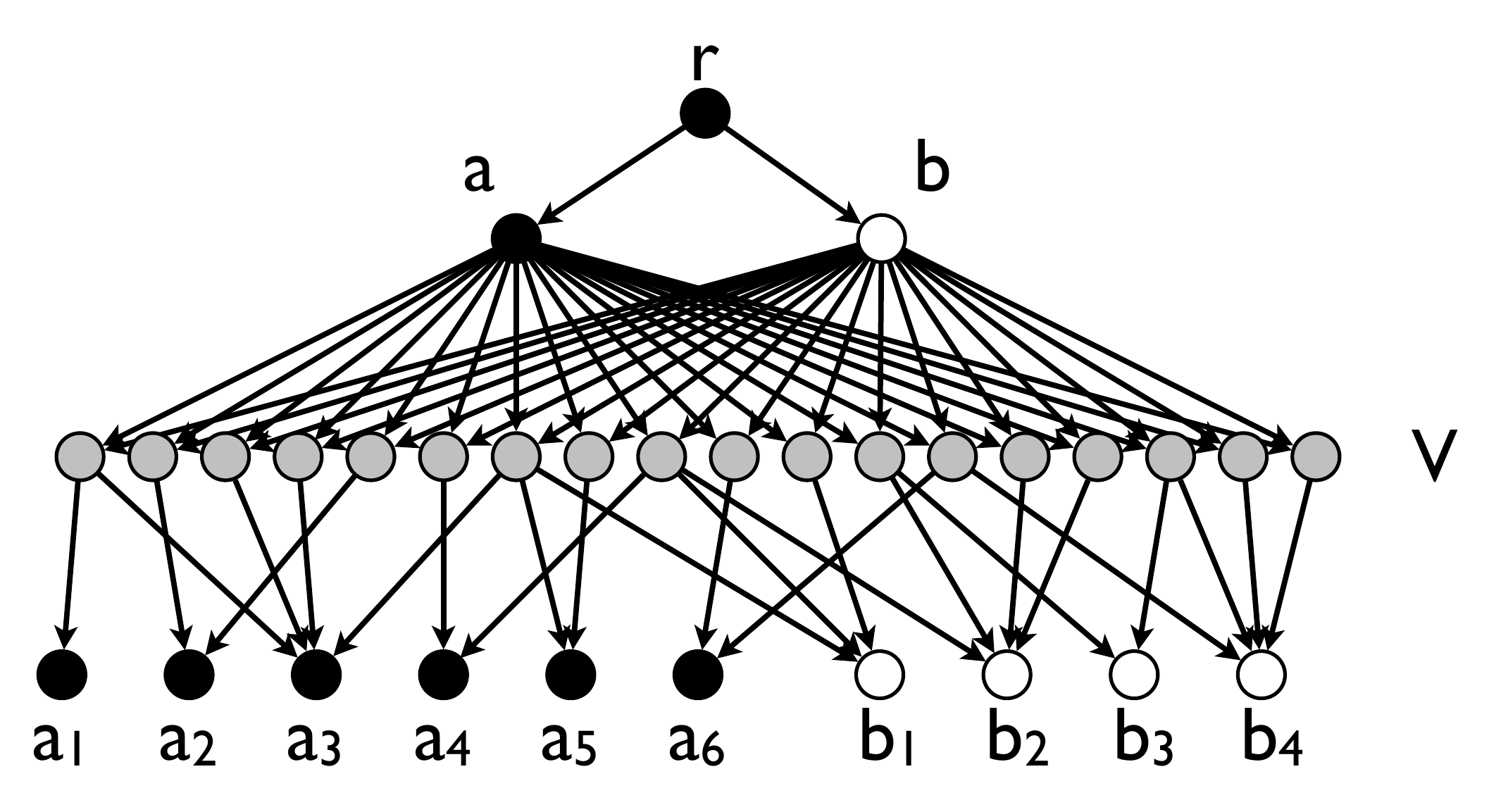}
}
\caption{A partially colored CRDAG constructed from a CNF theory
  $T$. The black $a_i$ nodes represent the positive clauses, the white
  $b_i$ nodes the negative clauses. The grey nodes represent the
  propositional variables; they are linked to the clauses in which
  they participate and have to be colored.}
\label{example-CRDAG}
\end{figure}

The problem being NP-complete, it is unlikely it can be solved by a
procedural program without search. The proof suggests the problem
becomes hard when nodes can have multiple parents. This situation is
not dealt with (and therefore usually abstracted away) in
traditional stemmatological methods. However, it does occur in the
datasets we analyzed. Constructing small examples where several nodes
have multiple parents, we quickly obtained an example for which the
procedural code erroneously claimed no connected coloring exists.  So
the worries of the developer about the completeness of the code were
grounded.



\subsection{An \fodotidp Solution}

A first \fodotidp solution used a binary relation \texttt{SameVariant}
for representing that two manuscripts have the same variant reading
and imposed two constraints: (i) transitivity of \texttt{SameVariant}
relation, (ii) manuscripts with the same variant reading have a common
ancestor with that variant reading and are connected to that ancestor
through manuscripts with that same variant reading. This resulted in a
working version that could serve as a golden standard for the
procedural code but was much slower than the latter.

As we already noticed in the shortest path problem, modeling
transitive closures results in large grounding sizes and
runtime. Hence, a major improvement can be expected when that can be
avoided. Representing the variant reading as a function from
manuscripts to variants allowed us to drop the transitivity
constraint.  The final improvement, resulting in the program below,
came from learning more about the procedural code: it checks for
connectedness by following a path to the original source manuscript of
the variant reading and checks that there is a single such source for
the variant reading. Expressing the latter as a single constraint
resulted in a version that turned out to be faster than the incomplete
procedural algorithm. The \idpdrie model is shown in
Listing~\ref{listing:idp-stemma} and explained below. 
We also show most of the procedural
code, so that the reader can see how a number of satisfiability-checking tasks
can be embedded in a single process.

\begin{lstlisting}[caption={Checking the consistency
    between stemma and features.},label={listing:idp-stemma}]
procedure main() {
  process("besoin")
  process("parzival")
  process("florilegium")
  process("sermon158")
  process("heinrichi")
}

/* ----- Knowledge base -------------------------------------- */
vocabulary V {
  type Manuscript
  type Variant
  CopiedBy(Manuscript,Manuscript)
  VariantReading(Manuscript): Variant
}
vocabulary Vtask {
  extern vocabulary V
  SourceOf(Variant): Manuscript
}
theory Ttask : Vtask {
  ! x : (x ~= SourceOf(VariantReading(x))) =>
     ? y : CopiedBy(y,x) & VariantReading(y) = VariantReading(x).
}

/* ----- Check consistency between feature and stemma -------- */
procedure check(feature) {
  setvocabulary(feature,Vtask)
  return sat(Ttask,feature)
}

/* ----- Procedures for processing --------------------------- */
procedure process(tradition) {
  io.write("Processing ",tradition,".\n")
  local path = "data/"
  local stemmafilename = path..tradition..".dot"
  local featurefilename = path..tradition..".json"
  processFiles(stemmafilename,featurefilename)
}
procedure processFiles(stemmafilename,featurefilename) {
  local stemma,nbnodes,nbedges = readStemma(stemmafilename)
  io.write("Stemma has ",nbnodes," nodes and
                       ",nbedges, " edges.\n")
  local nbp,nbs,time = processFeatures(stemma,featurefilename)
  io.write("Found ",nbp," positive out of ",nbs," groupings ")
  io.write("in ",time," sec.\n")
}
procedure readStemma(stemmafilename) {
  /* 19 lines of lua code */
}
procedure processFeatures(stemma,featurefilename) {
  /* 23 lines of lua code
    a loop iterating over the features, 
    -- compute feature as stemma extended with
        the feature specific data
    -- call check(feature)
    -- process the results
    finally, return the overall results */
}
\end{lstlisting}

The logical model is described in the ``Knowledge base'' section of the code.
The vocabulary has been split in two parts. The vocabulary
\texttt{V} is used to represent the input data: the stemma and the
feature. It introduces the types \texttt{Manuscript} and
\texttt{Variant}, the binary relation \texttt{CopiedBy} representing
the parent-child pairs in the given structure of the stemma and the
function \texttt{VariantReading} representing the known data about
variant readings of manuscripts.  The vocabulary \texttt{Vtask}
extends \texttt{V} with the task specific vocabulary. Only one extra
function is needed namely \texttt{sourceOf} which maps a variant
reading to the manuscript that is the source of that variant reading.
The theory \texttt{Ttask} consists of a single constraint; it states
that a manuscript that is not the source of its own variant reading
must have a parent with the same variant reading.

The remainder is procedural code. The procedure \texttt{main} iterates over all traditions to be analyzed and calls
the procedure \texttt{process} for each of those.
The latter procedure uses concatenation to construct two filenames
from the name of the tradition and passes these file names to the
\texttt{processFiles} procedure; the ``.dot'' file contains the stemma
data; the ``.json'' file the feature data. The \texttt{readStemma}
procedure (code omitted) returns the input structure describing the
stemma as well as the number of manuscripts (nodes) and parent-child
pairs (edges). The \texttt{processFeatures} procedure (code omitted)
iterates over the features in the file. For each feature, it
constructs a feature structure by extending the stemma structure with
the feature specific data. It then calls the \texttt{check}
procedure. This procedure extends the feature structure with the
symbols from the \texttt{Vtask} vocabulary
(\texttt{setvocabulary(feature,Vtask)}) and then checks the
color-connectedness of the feature (\texttt{sat(Ttask,feature)}).  The
yes/no result is returned to the \texttt{processFiles} procedure which
collects and returns the global results: number of consistent
(positive) features, total number of features and time. The
\texttt{processFiles} procedure prints these global data and returns
to \texttt{main}.

As can be seen in the {\tt main()} procedure, we used the code to
perform consistency checking for the features of 5 traditions; two of
them, Sermon 158 and Florilegium are real traditions, with stemmata
that have been constructed according to current philological best
practice; the other three are artificial traditions, produced under
test conditions by volunteers for the purposes of empirical research
into stemmatological methods.  We received the data from Tara
Andrews. A website where such stemma data can be found is
\url{http://byzantini.st/stemmaweb/}.  Some information about the
stemma we used is given in Table~\ref{tab:data}.

\begin{table}
\caption{The five traditions used in this work.}
\label{tab:data}
\centering{
\begin{tabular}{|l|l|l|l|l|l|}
\hline
Name &\# manu- & \# parent-child &\# features & \multicolumn{2}{c}{\# variant readings} \\
     &  scripts &pairs& & maximum & average \\
\hline
Notre Besoin & 13 & 13 & 44& 5  & 2,18 \\
Parzival & 21 & 20 & 122&  6 & 2,59 \\
Florilegium & 22 & 21 & 547& 5 & 2,19 \\
Sermon 158 & 34 & 33 & 270& 3  & 2,12 \\
Heinrichi & 48 & 51 & 1042& 17 & 4,84 \\
\hline
\end{tabular}
}
\end{table}

The \idp program determines consistency for all features and datasets
in a matter of seconds\footnote{Using an Intel$^R$ Core$^{TM}$2 Duo
  CPU at 3.00GHz with 3.7 GB of RAM running Ubuntu with the \idpdrie
  options stdoptions.groundwithbounds = false (disabling bounded
  grounding) and stdoptions.liftedunitpropagation = false (disabling
  lifted unit propagation).}:

\begin{footnotesize}
\begin{verbatim}
> main()
Processing besoin.
Stemma has 13 nodes and 13 edges.
Found 26 positive out of 44 groupings in 0 sec.
Processing parzival.
Stemma has 21 nodes and 20 edges.
Found 45 positive out of 122 groupings in 1 sec.
Processing florilegium.
Stemma has 22 nodes and 21 edges.
Found 431 positive out of 547 groupings in 2 sec.
Processing sermon158.
Stemma has 34 nodes and 33 edges.
Found 64 positive out of 270 groupings in 2 sec.
Processing heinrichi.
Stemma has 48 nodes and 51 edges.
Found 1 positive out of 1042 groupings in 12 sec.
>
\end{verbatim}
\end{footnotesize}

Our largest benchmark is the  heinrichi
data set \cite{RoosH/09}. This stemma about old Finnish texts includes 48
manuscripts, 51 \texttt{copiedBy} tuples and information about 1042
features. Processing all features takes 12 seconds with the IDP system while
it took 25 seconds with the original procedural code.

One can observe that rather few features are consistent
with the stemma. This raises the question what is the minimal number
of sources needed to explain the data. To solve that inference task, it suffices to
replace the vocabulary extension \texttt{Vtask} and the theory
\texttt{Ttask} in the knowledge base and to introduce the term to
be minimized. As core procedure, \texttt{Check} is replaced by
\texttt{minSources} and the processing of results has to be
adjusted. The most relevant new parts are shown in
Listing~\ref{listing:ms}. The \texttt{IsSource} predicate is
defined as manuscripts that do not have a parent with the same variant
reading.

\begin{lstlisting}[caption={Minimize the number of sources.},label={listing:ms}]
/* ----- new parts of Knowledge base ------------------------- */
vocabulary Vms {
  extern vocabulary V
  IsSource(Manuscript)
}
theory Tms : Vms {
  {! x : IsSource(x) <- ~? y : CopiedBy(y,x) &
                        VariantReading(y) = VariantReading(x).}
}
term NbOfSources : Vms {
  #{ x : IsSource(x) }
}

/* ----- the core procedure ---------------------------------- */
procedure minSources(feature) {
  setvocabulary(feature,Vms)
  return minimize(Tms,feature,NbOfSources)[1]
}
\end{lstlisting}

Although this is a minimization problem, processing the traditions is
still a matter of seconds, except for the larger Heinrichi dataset
which now requires about 5 minutes to process its 1042 features. 

Other variations are of interest to the researchers. One variation, mentioned by
\citeN{Andrewsetal12}, considers the possibility that the
scribe has copied from an older ancestor than the direct parent, thus
reintroducing a variant. Playing with the relative penalty
of introducing a new variant versus reverting to an older variant, one
can obtain various explanations of interest to the stemmatologist. All
these can be achieved with modifying a handful of lines in the
model. Interesting about the above variant is that it uses a predicate
\texttt{IndirectAncestor} that is defined in terms of the stemma data,
so it can be computed once and reused when processing each of the
features. As illustrated in
Listing~\ref{listing:revert}, the tight integration of the knowledge base with the
procedural code makes this very easy\footnote{With a more recent version of
  \idpdrie, the user can leave this optimization to the system \cite{tplp/JansenJJ13}.}.
The procedure \texttt{readStemma}, which
constructs the stemma structure from the inputfile, is extended with the
call \texttt{modelexpand(T,stemma)[1]}. The resulting model is the
stemma structure extended with the true \texttt{IndirectAncestor}
atoms. This structure, together with the other outputs of
\texttt{readStemma}, is returned to the procedure
\texttt{processFiles} which uses it to handle the features one by one.

\begin{lstlisting}[caption={Materializing a definition once and using
    the materialization many times.},label={listing:revert}]
vocabulary V {
  /* ... as in Listing 5 ... */
  IndirectAncestor(Manuscript,Manuscript)
}

theory T : V {
  {! x y : IndirectAncestor(x,y) <-
            ? z : CopiedBy(x,z) & IndirectAncestor(z,y).
   ! x y : IndirectAncestor(x,y) <-
            ? z : CopiedBy(x,z) & CopiedBy(z,y).}
}

procedure readStemma(stemmafilename) {
  local stemma = newstructure(V,"stemma")
  /* ... reading the stemma data ... */
  return modelexpand(T,stemma)[1], #nodes, #edges
}
\end{lstlisting}

\section{Minimum common supergraphs of partially labeled trees}\label{sec:CS}

\emph{Phylogenetic trees}, extensively surveyed by
\citeN{felsenstein-inferring}, are the traditional tool for
representing the evolution of a given set of species.  However, there
exist situations in which a tree representation is inadequate. One
reason is the presence of evolutionary events that cannot be displayed
by a tree: genes may be duplicated, transferred or lost, and
recombination events (i.e., the breaking of a DNA strand followed by
its reinsertion into a different DNA molecule) as well as
hybridization events (i.e., the combination of genetic material from
several species) are known to occur. A second reason is that even when
evolution is indeed tree-like, there are cases in which a relatively
large number of tree topologies are ``equally good'' according to the
chosen criterion, and that not enough information is available to
discriminate between those trees. One solution that has been proposed
to address the latter issue is the use of \emph{consensus trees},
where the idea is to find a tree that represents a compromise between
the given topologies. Another approach, the focus of this section,
consists in building a network that is compatible with all topologies
of interest. A somewhat loose description of the variant we are
interested in, which will be stated in a more formal way below, is to
find the smallest graph that contains a given set of evolutionary
trees.  For more information about \emph{phylogenetic networks}, see
the recent book by \citeN{huson-networks-book} and the online,
up-to-date annotated bibliography maintained by \citeN{gambette-who}.

\subsection{The problem}

The studied problem is about the evolution of a fixed set of $m$
species. The input is a set of phylogenetic trees, each tree showing a
plausible relationship between the $m$ species.  All trees have $n$
($>m$) nodes, $m$ of them are labeled with the name of the species
(typically, in the leaves, but also internal nodes can be
labeled). Given $n-m$ extra names, the labeling of each tree can be
extended into a full labeling. Now, we can consider the union of these
full labelings: a network with $m$ labeled nodes and edges which are
induced by the bijections between the fully labeled trees and the
network. Obviously, the number of edges of the network depends on the
chosen full labelings of the trees.  The task is to find a network
with a minimum number of edges.  Below, we formulate the problem as a
slightly more general graph problem where we do not fix the size of
the initial labeling.

\begin{definition}[Common supergraph of partially labeled $n$-graphs]
  Given is a set $N$ of $n$ names and a set of graphs
  $\{G_1,G_2,\ldots,G_t\}$ where each graph $G_i= (V, E_i, \mathcal
  L_i)$ has $n$ vertices (the set $V$), edges connecting pairs of
  vertices (the set $E_i$) and where some of the vertices are labeled by
  names (an injective partial function $\mathcal L_i: V \rightarrow
  N$).
  A graph $(N,EN)$ is a {\em common supergraph} of
  $\{G_1,G_2,\ldots,G_t\}$ if there exists, for each $i$, a bijection
  $\mathcal L'_i: V \rightarrow N$ that extends $\mathcal L_i$ and
  such that 
  $\{v,w\} \in EN$ iff there exists an $i$ such that $\{v',w'\} \in
  E_i$ and $\{v,w\} =\{(\mathcal L'_i(v'),  \mathcal L'_i(w')) \}$.

  A common supergraph $(N,EN)$ is a {\em minimum} common supergraph if
  no other common supergraph $(N,EN')$ exists for which $|EN'| < |EN|$.
\end{definition}

Note that every labeling function $\mathcal L'_i$ induces an injection
$E_i\to EN$, hence the name common supergraph.
Figure~\ref{fig:n-k-trees-mcs-example} shows two partially labeled
$7$-graphs, along with two of their common supergraphs. $G_2$ is not a
minimum common supergraph since it has more edges than $G_1$; $G_1$ is
a minimum common supergraph since $T_1$ and $T_2$ are not
isomorphic and $G_1$ has only one more edge than each of $T_1$ and
$T_2$.

\begin{figure*}[htbp]
\begin{center}
\scalebox{0.65}{
\begin{tabular}{cccc}
\begin{tikzpicture}[>=stealth]
\tikzstyle{every node}=[draw,circle,inner sep=2pt];
    \node (1) at (0,-1) {$1$};
    \node (v1) at (1,-1)  {};
    \node (v2) at (2,-1)  {};
    \node (3) at (3,-1) {$3$};
    \node (4) at (1,-2)  {$4$};
    \node (v3) at (1,-3)  {};
    \node (2) at (1,-4)  {$2$};
    \draw  (1) -- (v1) -- (v2) -- (3);
    \draw (v1) --  (4) -- (v3) -- (2);
\end{tikzpicture}
&
\begin{tikzpicture}[>=stealth]
\tikzstyle{every node}=[draw,circle,inner sep=2pt];
    \node (1) at (0,-1) {$1$};
    \node (v1) at (1,-1)  {};
    \node (v2) at (2,-1)  {};
    \node (3) at (3,-1) {$3$};
    \node (4) at (2,-2)  {$4$};
    \node (v3) at (2,-3)  {};
    \node (2) at (2,-4)  {$2$};
    \draw  (1) -- (v1) -- (v2) -- (3);
    \draw (v2) --  (4) -- (v3) -- (2);
\end{tikzpicture}
&
\begin{tikzpicture}[>=stealth]
\tikzstyle{every node}=[draw,circle,inner sep=2pt];
    \node (1) at (0,-1) {$1$};
    \node (v1) at (1,-1)  {};
    \node (v2) at (2,-1)  {};
    \node (3) at (3,-1) {$3$};
    \node (4) at (1.5,-2)  {$4$};
    \node (v3) at (1.5,-3)  {};
    \node (2) at (1.5,-4)  {$2$};
    \draw  (1) -- (v1) -- (v2) -- (3);
    \draw (v1) --  (4) -- (v3) -- (2);
    \draw (v2) -- (4);
\end{tikzpicture}
&
\begin{tikzpicture}[>=stealth]
\tikzstyle{every node}=[draw,circle,inner sep=2pt];
    \node (1) at (-1,-1) {$1$};
    \node (v1) at (0,-1)  {};
    \node (v2) at (0,-3)  {};
    \node (3) at (0,-4) {$3$};
    \node (4) at (0,-2)  {$4$};
    \node (v3) at (-1,-2)  {};
    \node (2) at (1,-2)  {$2$};
    \draw (1) -- (v1)  -- (2) -- (v2) -- (v3) -- (3);
    \draw (v1) -- (4) -- (v2);
    \draw (1) -- (v2);
    \draw (v1) -- (v3) -- (4);
\end{tikzpicture}
 \\
      &       &       &       \\
$T_1\quad\quad\quad$ & $\quad\quad\quad T_2$ & $G_1$ & $G_2\quad\quad$
\end{tabular}
}
\end{center}
\caption{Two $7$-graphs, $T_1$ and $T_2$, and two of their common supergraphs.
  $G_1$ is a minimum common supergraph.}
\label{fig:n-k-trees-mcs-example}
\end{figure*}
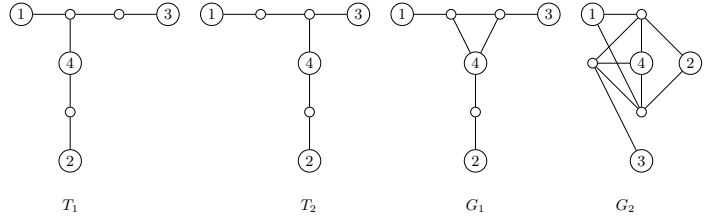

Now, we can consider the following decision problem: Given a set of
partially labeled $n$-graphs, can the
labelings 
be completed such that
the $n$-graphs have a common supergraph with at most $k$ edges?
\citeN{ump} prove that this problem is NP-hard, even if the $n$-graphs
are trees with all leaves labeled.

\subsection{An \fodotidp solution}

Listing~\ref{fig:idp-model-adv} shows a simple model inspired by
\cite{ump}. It makes use of three types, \texttt{tree},
\texttt{vertex}, and \texttt{name}. The latter two types have the same
number of elements in a correct input structure. The structure of the
given trees is described by the ternary predicate \texttt{edge} (the
first argument refers to the tree to which the edge belongs), the structure of
the common supergraph (over the names themselves) by the predicate
\texttt{arc}. The labeling is described by the function \texttt{label}
from the nodes of the given trees to the names. It is partially given
in the input structure and it is completed during model expansion. The
constraint in the theory, stating that, for each name \texttt{nm} and
each tree \texttt{t}, there exists {\em exactly one} node \texttt{nd}
(denoted \idpcode{?  1 nd}) such that its label is \texttt{nm},
ensures that the labeling is bijective. The \texttt{arc} atoms can be
defined as the pairs of names induced by the labels on the nodes of an
edge of the tree (the definition in the theory). However, as the
minimization is on the number of \texttt{arc} atoms in a model, some
care is required. One should ensure either that \texttt{arc} is a
symmetric relation or that there is at most one \texttt{arc} atom for
each pair of names. The latter approach is taken as it gives a
somewhat smaller grounding. It is achieved by exploiting the total
order which exists over each domain (the tests \idpcode{label(t,x) <
  label(t,y)}).

\begin{lstlisting}[caption={Modeling {\sc cs-plt} in \fodotidp.},label={fig:idp-model-adv}]
vocabulary CsPltVoc {
  type tree
  type vertex
  type name // Isomorphic to vertex
  edge(tree,node,node)  // trees, given in input structure
  arc(name,name) // the induced network
  label(tree,node): name // the labeling,
                    // partially given in the input structure
}
theory CsPltTheory : CsPltVoc {
  { // induced network; arc is anti-symmetric
   ! t x y : arc(label(t,x),label(t,y)) <- edge(t,x,y) &
                                  label(t,x) < label(t,y).
   ! t x y : arc(label(t,x),label(t,y)) <- edge(t,y,x) &
                                  label(t,x) < label(t,y).
  }
  ! t nm : ?1 nd : label(t,nd) = nm. // label is bijective
}
term SizeOfSupergraph : CsPltVoc { #{ x y : arc(x,y) } }
procedure main() {
  print(minimize(CsPltTheory,CsPltStructure,SizeOfSupergraph)[1])
}
\end{lstlisting}

The rules of the \texttt{arc} definition in this solution each have
have two occurrences of the terms \idpcode{label(t,x)} and
\idpcode{label(t,y)}. The current grounder naively associates a
distinct symbol with each occurrence, which boils down to grounding a
clause of the following form:
\begin{lstlisting}[caption={One of the  \texttt{arc} rules after
    initial processing by the grounder.},label={fig:expandedarc}]
  ! t x y lx1 lx2 ly1 ly2 : arc(lx1,ly1) <- lx1=label(t,x) &
        ly1=label(t,y) & lx2=label(t,x) & ly2=label(t,y) &
        edge(t,x,y) & lx2 < ly2.
\end{lstlisting}

This approach creates extra variables and very large groundings. To
avoid this behavior, on can rewrite the definition as:
\begin{lstlisting}[caption={Better performing definition of
    \texttt{arc}.},label={fig:betterarc}]
{ // induced network
 ! t x y lx ly : arc(lx,ly) <- lx=label(t,x) & ly=label(t,y) &
                                  edge(t,x,y) & lx < ly.
 ! t x y lx ly : arc(lx,ly) <- lx=label(t,x) & ly=label(t,y) &
                                  edge(t,y,x) & lx < ly.
}
\end{lstlisting}
While the formulation is less elegant, the effect on the size of the
grounding and the solving time is dramatic; e.g., the grounding is
reduced from 620798 to 6024 propositional clauses and the solving
time from 144s to 8s on a problem with 5 trees of 8 vertices and 4
initial labels.

One can explore several other variations. As mentioned above, one
could use a symmetric \texttt{arc} relation. Also, as the \texttt{arc}
definition is free of recursion, one could replace it with the two
implications of the completion. Then, exploiting the minimization on
the number of \texttt{arc} atoms, one could drop the only-if part of
the completion.  The effect on solving time of all these variations is
rather marginal.

\subsection{An approximate solution}

The solving time is exponential in the number of nodes and, if several
trees are involved, the program becomes impractical on real-world
problems, even if the best solution found so far is returned when some
time budget is exceeded. However, the versatility of the IDP system
allowed us to experiment with various strategies for greedily
searching an approximate solution.  This led to the following quite
natural solution that performed very well, with respect to both
running time and quality of the solution.
\begin{enumerate}
\item Find a minimum common supergraph (MCS) for every pair of trees.
\item Pick an MCS with minimum size (say $G$) and remove the two trees
   that are the input for $G$.
\item Find an MCS between $G$ and every remaining tree.
\item Replace $G$\footnote{This way, the MCS is assembled by each time incorporating one additional original tree.} by an MCS with minimum size, remove the tree
  that is the input for this MCS and go back to step 3 if any tree remains.
\end{enumerate}

Steps 1 and 3 of this simple procedure are performed by \idpdrie using a
model very similar to that of Listing~\ref{fig:idp-model-adv} (see
\citeN{ump} for the actual model)\footnote{The whole method can be
  implemented as an \idpdrie procedure; however, the scripts had been
  implemented before the Lua interface was available.}. This greedy
approach works very well. Indeed, for large instances and a fixed time
budget, the exact method runs out of time and returns a suboptimal
solution while the greedy method completes and returns a solution
that, although suboptimal, is typically much
smaller. \Cref{tab:ump-results} shows some experimental results on
randomly generated data with various parameters. A timeout was set to
2000s\footnote{Using an Intel$^R$ Core$^{TM}$ i7 CPU 870 at 2.93GHz
   with 8GB of RAM running Ubuntu; default settings for \idpdrie.}
and average number of edges were recorded over four runs for each
instance for both the exact and the greedy method.

\begin{table}[tb]
 \scalebox{0.85}{
\begin{tabular}{|ccc|ll|}
          &         & \#initial  &  exact & greedy\\
  \#trees & \#nodes & labels & \# edges & \# edges\\ \hline
  5  & 55   &  5&\bf{130} & 131.25\\
  5  & 60   & 10& \bf{128} & 132.75\\
  5  & 75   & 25 &207.75 * & \bf{184.75}\\
 10 &  55   &  5 & 183.75 * & \bf{154.50}\\
 10 &  60  &  10 & 177.75 * & \bf{154.75}\\
 10  & 75   & 25 & 270.00 * & \bf{269.25}\\
 20  & 55   &  5 & 241.50 * & \bf{171.75}\\
 20  & 60  &  10 & 232.00 * & \bf{152.25}\\
 20  & 75 &   25 & 346.25 * & \bf{279.00}\\
\end{tabular}
}
\caption{Randomly generated instances of the minimum common subgraph
  problem solved with a time bound of 2000s. Sizes of MCS (average
  over four runs) for exact and greedy approach. *: approximate
  solution due to time out.}
\label{tab:ump-results}
\end{table}

\section{Learning deterministic finite state automata}\label{sec:dfa}

A third task is about learning a {\em deterministic finite state
  automaton} (DFA). The goal is to find a (non-unique) smallest DFA
that is consistent with a given set of positive and negative
examples. It is one of the best studied problems in grammatical
inference \cite{pr/Higuera05}, has many application areas, and is
known to be NP-complete~\cite{iandc/Gold78}.
Interestingly, one of the first algorithms proposed to solve this
problem was based on a translation to constraint programming
~\cite{biermannfeldman72}. Much later, translations of this problem to
graph coloring~\cite{CosteN97,costaverwer12} and
satisfiability~\cite{grinchtein06,Verwer} were
proposed. Although the DFA learning problem is typically tackled using
greedy approaches \cite{pr/Higuera05},~\citeN{Verwer12} recently won the
2010 Stamina DFA learning competition~\cite{stamina} by an improved
translation to a SAT problem and running an off-the-shelf SAT
solver. Here we explore to what extent an \fodotidp formalization can
compete with this competition winner.

\subsection{The problem}

A \emph{deterministic finite state automaton} (DFA) is a directed
graph consisting of a set of \emph{states} $Q$ (nodes) and labeled
\emph{transitions} $T$ (directed edges). The root is the start state
and any state is either an \emph{accepting} or a \emph{rejecting}
state. In each state, there is exactly one transition for each
symbol. A DFA defines a language, the set of strings it accepts. It
can be used to \emph{generate} or \emph{verify}
sequences of symbols (strings) using a process called \emph{DFA
  computation}. When verifying strings, the symbols of the input string
determine a path through the graph. When the final state is an
accepting state, the string is accepted, otherwise it is rejected.

Given a pair of finite sets of positive example strings $S^{+}$ and
negative example strings $S^{-}$, (the \emph{input sample}), the goal
of \emph{DFA identification} (or \emph{learning}) is to find a
(non-unique) \emph{smallest} DFA $\mathit{A}$ that is
\emph{consistent} with $S = \{S^+, S^-\}$, i.e., every string in $S^+$
is accepted, and every string in $S^-$ is rejected by
$\mathit{A}$. Typically, the size of a DFA is measured by $|Q|$, the
number of states it contains.

Most DFA learning algorithms are based on the method of state-merging.
This method first constructs a tree-shaped automaton called the
\emph{augmented prefix tree acceptor} (APTA).  As can be seen in
Figure~\ref{fig:apta}, the APTA accepts the positive examples and
rejects the negative ones. Other strings either end up in a non-final
state or cannot be processed due to a missing transition. 
The APTA automaton can be \emph{completed} to obtain a DFA with the same
number of states by (arbitrarily) labeling the non-final states and
adding the missing transitions. When all non-final states are labeled as
reject and all extra transitions target a reject state with no path to 
an accepting state, this DFA accepts only the positive examples.

A smaller DFA, accepting more strings, can be constructed by
state-merging on the APTA.  Merging states under the constraints that
the automaton remains deterministic (at most one transition/label in
each state) and that accepting and rejecting states cannot be merged
preserves consistency with the input sample. State-merging increases
the number of strings accepted by the automaton, and hence generalizes
the language accepted by the DFA that completes the automaton.

\begin{figure}[t]
\begin{center}
\scalebox{0.60}{
\includegraphics{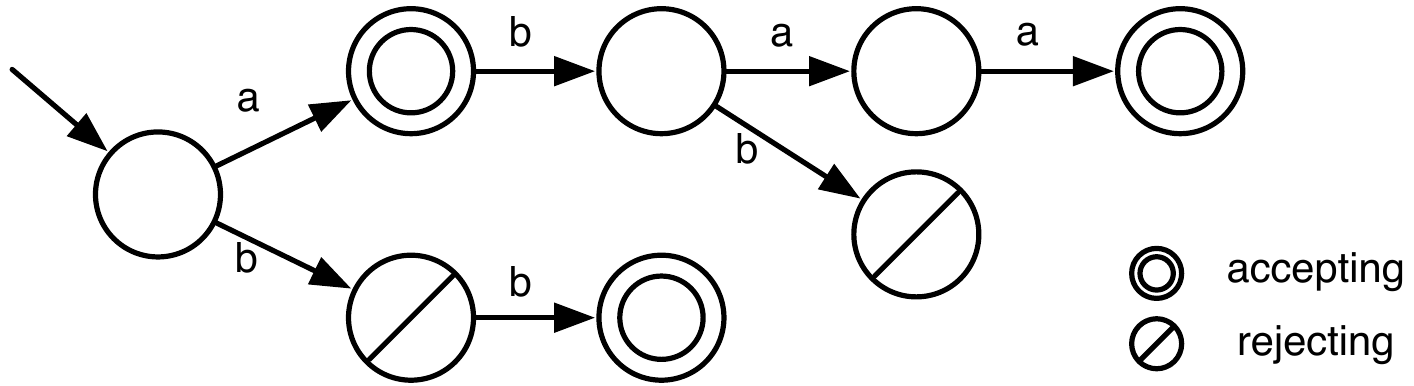}
}
\caption{An augmented prefix tree acceptor (APTA) for $S = (S^+ = \{a,
  abaa, bb\}, S^- = \{abb,b\})$. The start state (annotated with
  incoming arrow) is the root of the APTA.  }
\label{fig:apta}
\end{center}
\end{figure}

States of the final automaton are thus equivalence classes of states
of the APTA.  Calling the states of the final automaton colors, the
problem becomes that of finding a coloring of the states of the APTA
that is consistent with the input sample. Following \citeN{CosteN97},
\citeN{Verwer} take this approach. They formulate constraints
expressing which pairs of states are incompatible, and abstract the
problem as a graph.  The nodes of this graph are the states of the
APTA and the edges are the incompatible pairs. The decision problem,
whether there exists an automaton with $k$ states, becomes a graph
coloring problem for $k$ colors. They use a clever SAT encoding to
solve this decision problem and embed it in a workflow to solve the
minimization problem.  For really large problems, the SAT formulation
becomes too big (hundreds of colors, resulting in over 100 million
clauses) to be handled by a SAT solver \cite{Verwer}. To reduce the
problem size, they used a greedy heuristic procedural method based on
state-merging. Every merge performed by this method reduces the size
of the APTA and therefore also the size of the encoding. In addition,
this preprocessing identifies a clique of pairwise incompatible states
in the APTA. For states in this clique, the colors can be fixed in
advance. The effect is to break the symmetries between these colors
and thus to further reduce the size of the problem.
The preprocessing also deduces that certain
state/color combinations cannot result in a solution.  A preprocessed
problem instance is then extended with a set of SAT clauses and is the
input for the SAT solver. The SAT clauses express the constraints of
the problem and are generated from the instance.

\subsection{An \fodotidp solution}

Our goal is not to set up the complete workflow described above, but
to compare the performance of the native SAT encoding of
\citeN{Verwer} and \citeN{Verwer12} with the performance of an
\fodotidp model on the same problem instances as obtained after the
preprocessing. Our \fodotidp model for solving a single instance is
shown in Listing~\ref{fig:idp-model-dfa}.

\begin{lstlisting}[caption={Modeling DFA in \fodotidp.},label={fig:idp-model-dfa}]
vocabulary dfaVoc {
  type state // states used in APTA
  type label // symbols triggering transitions
  type color // available states for resulting automaton
  partial trans(state,label): state // transitions of APTA
  acc(state) // accepting states of APTA
  rej(state) // rejecting states of APTA
  colorOf(state): color // fixed in input for colors in clique
  // the resulting automaton:
  partial colorTrans(color,label): color // transitions of DFA
  accColor(color) // accepting states
}
theory dfaTheory : dfaVoc {
  ! x : acc(x) => accColor(colorOf(x)).
  ! x : rej(x) => ~accColor(colorOf(x)).
  // trans induces colorTrans:
  ! x l z: trans(x,l)=z => 
       ! i j : colorOf(x)=i & colorOf(z)=j => colorTrans(i,l)=j.
}
procedure main() {
  print(modelexpand(dfaTheory,instance)[1])
}
\end{lstlisting}

The types \texttt{state}, \texttt{label}, the function \texttt{trans},
and the predicates \texttt{acc} and \texttt{rej} describe the given
input samples (and hence the APTA). Note that \texttt{trans} is
partial as it is only defined for the transitions present in the input
sample.  The states of the resulting automaton are elements of the
type \texttt{color}.  Its transitions are described by the function
\texttt{colorTrans}. This function is also declared as a partial
function. To obtain a complete DFA, the function \texttt{colorTrans}
has to be extended with the missing transitions. Which one is assigned
does not matter as it does not affect the processing of the strings in
the input sample (though it has an effect on the language that is
accepted).  The function \texttt{colorOf} maps the states of the APTA
on the states (colors) of the DFA. The predicate \texttt{accColor}
describes the accepting states of the resulting automaton.

The theory expresses two constraints on \texttt{accColor}: accepting
states of the APTA must and rejecting states cannot be mapped to an
accepting state of the DFA. A third constraint states that each
transition in the APTA induces a transition (between colors) in the
DFA. 

The input structure, which is omitted, not only completely defines the
types \texttt{color}, \texttt{state} and \texttt{label} but also the
APTA. That implies that the \idpdrie grounder has complete
knowledge about the relations \texttt{accept} and \texttt{reject} and
the function \texttt{trans}.  Hence, for example, the grounder only
grounds the formula\\ \idpcode{colorTrans(colorOf(x),l)=colorOf(z)} for
tuples \texttt{(x,l,z)} for which \idpcode{trans(x,l)=z} is true
in the input structure. Further, the input structure also contains the
partial information about \texttt{colorOf} and \texttt{colorTrans}
that has been derived by the preprocessing.

The \texttt{main} procedure assumes that the input structure is named
\texttt{instance}; it calls the solver to search for a model and
prints it.

The above model is a very natural formulation of the problem and 
corresponds quite closely to a ``decompilation'' of the SAT
clauses expressing the constraints of the problem (Table 1 in both
\citeN{Verwer} and \citeN{Verwer12}).  The most noticeable difference is in a
redundant constraint which can be decompiled into:
\begin{lstlisting}[caption={Redundant constraint},label={fig:redundant-dfa}]
!v l w: trans(w,l)=v => 
             !j: colorTrans(colorOf(w),l)=j => colorOf(v)=j.
\end{lstlisting} 
This formula can be derived from the last formula in our theory
together with the fact that \texttt{colorOf} is a total function.

In \fodotidp it is very straightforward to extend the model with a term
counting the number of colors used and to minimize that number. This
makes our \fodotidp method very similar to the optimization method used
by \citeN{Verwer,Verwer12}. \idp, however, has several advantages over
an encoding constructed by hand. For instance, variants of the model
such as minimizing the number of
transitions in the DFA instead of the number of states are very
straightforward to obtain, but would require a major reengineering of
the SAT encoding.  This has practical value as different application
domains of DFA prefer different optimization criteria. Furthermore, it
is much easier to introduce bugs in handmade encodings than in the few
lines of \idp code. In fact, by analyzing and comparing the results of \idp
and the handmade encoding, we discovered some subtle bugs in 
the handmade translation that caused an incorrect answer in rare occasions.

\subsection{Experiments}


\begin{figure}[t]
\begin{center}
\scalebox{0.75}{
\includegraphics{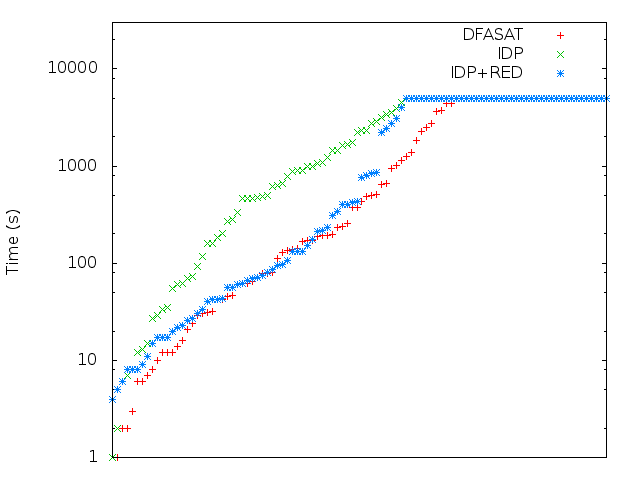}
}
\caption{Solving time for the SAT encoding (DFASAT), for the \fodotidp model
  of Listing~\ref{fig:idp-model-dfa} (IDP) and for that model extended
  with the redundant constraint of Listing~\ref{fig:redundant-dfa}
  (IDP+RED). Times are monotonically increasing, so the order of
  problem instances is different for each system. Timeout is set at
  5000s. 69 problems are solved by DFASAT, 59 by IDP and IDP+RED.}
\label{fig:solvingtime}
\end{center}
\end{figure}


We compared the performance of our model with that of the SAT encoding, denoted DFASAT, for 100 tough problems (the DFA is restricted to have only 5 states on top of those in the initial clique) from the 2010 Stamina DFA learning competition~\cite{stamina}.\footnote{Using an Intel$^R$ Core$^{TM}$ i5-2500 CPU at 3.30GHz with 7.7 GB of RAM Running Ubuntu. Memory use was limited to 4GB, time to 5000 seconds. \idpdrie was ran with standard options.} 

Figure~\ref{fig:solvingtime} compares the solving time of DFASAT with
that of two \fodotidp models. The first one is as shown in
Listing~\ref{fig:idp-model-dfa} (IDP); the second one extends the
model with the redundant constraint of Listing~\ref{fig:redundant-dfa}
(IDP+RED). One can observe that the redundant constraint improves the
performance of the \idpdrie system and that the performance comes
quite close to that of DFASAT. Still, DFASAT can solve more problems
than IDP+RED (69 vs.\ 59).  We also have to add that the
dedicated preprocessing that generates the SAT instances requires on
average 5 seconds while the grounding takes substantially more
time. For the IDP version, it is on average 124s; for ID+RED, the
average is 168s.
The results reported here are substantially better than those reported
in~\cite{iclp/Blockeeletall12}. By analyzing these earlier results, we
unraveled that the large performance gap was due to our grounding
being three times the size of the SAT encoding. This was caused by the
introduction of unneeded auxiliary predicates (so called Tseitins)
during the grounding. The problem was repaired in a new version of the
grounder. As mentioned before, our detailed
analysis also revealed subtle bugs in the dedicated preprocessing 
which generates the SAT encodings for DFASAT.
Thus, using an \idp implementation of a hard problem such as DFA
learning and comparing it to a fast competition winning SAT translation
was not only useful for improving \idp, but also for improving the 
competition winner.

It is very encouraging to observe that the performance of a tiny \emph{and}
comprehensible predicate logic model comes very close to that of an
ingeniously-tuned SAT-encoding that is a key component of a
competition winner.

\section{Conclusion}\label{sec:concl}

In this paper, we presented the \idpdrie system from a user's
perspective. We introduced the various components of an \fodotidp
model and illustrated their use in a model for the shortest path
problem. We also showed models for some problems encountered by
researchers in data mining and machine learning.  In a first problem
from stemmatology, \fodotidp models proved to be of invaluable help
for researchers trying to cope with stemma that go beyond tree
structures \cite{Andrewsetal12}. We obtained a model that not only
correctly handles arbitrary directed acyclic graphs, but also achieved
better performance than the original (incomplete) procedural code. In
the second problem, about phylogenetic trees, \fodotidp models helped
researchers to explore approximate solutions for an NP-hard
problem~\cite{ump}. The third problem we modeled is the classical
problem of learning a deterministic finite state automaton. We
compared an \fodotidp model with a state of the art SAT encoding of
the problem. Here we found that the performance of an \idpdrie
solution comes pretty close to that of a highly tuned SAT encoding.
These applications illustrate that \fodotidp models are a valuable
alternative for dedicated procedural code when novel data needs to be
analyzed and explored. Interestingly, in both problems where we
compared with an existing solution (stemmatology and DFA learning), we
uncovered some bugs in those solutions. It is fair to add that we also
uncovered some cases where the grounder of the \idpdrie system performed 
a suboptimal job. In the minimum common
supergraph application we found it deals poorly with multiple
occurrences of the same term in a formula; in the DFA application we
found that it introduced unneeded auxiliary symbols. While the latter
problem has already been solved, the former is, at the time of
writing, still on the todo list of the implementation team.


Our work is a further indication that the \idpdrie system is coming of
age. It was already known from the ASP-competitions that it compares
pretty well with ASP systems in terms of performance
\cite{lpnmr/DeneckerVBGT09,LPNRM/Calimeri11}. In contrast to ASP,
which relies on the stable semantics \cite{iclp/GelfondL88}, it is
based on first-order logic. The informal semantics of FO's connectives
and of the novel language constructs is clear and easy to understand.
This probably makes it easier for newcomers to start modeling.  For
example, the authors of the minimum common supergraph problem
\cite{ump} were neither familiar with Prolog nor with \fodotidp and
hardly needed any help from the IDP team.
The core of an \fodotidp model consists on the one hand of formulas in
first-order logic, which act as constraints, and on the other hand of
definitions, which are close to the rules of traditional logic
programs. Given interpretations for open predicates (the predicates
that are not defined in the theory), the definitions determine a
unique model through the well-founded semantics \cite{GelderRS91}. The
search results in an interpretation of the open predicates and hence a
model of the theory that is consistent with the constraints.  What
distinguishes \fodot from traditional logic programming is the use of
non-Herbrand interpretations and correspondingly, the lack of
constructor functions. This often leads to a simpler data
representation and gives rise to elegant model formulations. On the
other hand, there are cases where the rich data structures that arise
in Herbrand interpretations (compound terms, lists, trees, \ldots) are
useful too and these currently cannot easily be modeled in IDP3.
Another distinction is that the IDP framework offers other forms of
inference, most notably model expansion and model minimization. A
feature of the \idpdrie system is the integration of procedures in
\fodotidp models~\cite{inap/DePooterWD11} and the clean separation
between declarative and procedural components. As we illustrated in
the stemmatology application, this allows a user to develop a whole
workflow in an \fodotidp model.

The logic of \fodotidp extends predicate logic with inductive
definitions, types, arithmetic, aggregates and partial functions. Of
these, inductive definitions is the most fundamental one. The basis of
the language is predicate logic. In fact, in many applications, the
extensions only serve for making models more readable. For example,
the aggregates (in the form of quantifications $?<2$) in the shortest
path problem are directly translatable to FO and the (non-recursive)
definition in the minimum common supergraph problem is equivalent with
its completion.  Hence, three out of the four problems we describe in
this text are in fact solved with pure predicate logic models.

Our work on applications taught us also a few things about good
models. In all problems we solved in this paper, a class of objects is
separated in equivalence classes. (In the shortest path problem there
is the class of edges participating in the path, and the class of
other edges). It is tempting to represent these equivalence classes by
the transitive closure of some relation. However, the transitive
closure of a binary relation is expensive. It gives rise to large
groundings and this, together with the cost of checking for unfounded
sets, results in poor performance.  Binary transitive closures arise
naturally during modeling but they are better avoided in the IDP3
system. In the shortest path problem our first solution had a binary
transitive \texttt{reaches} relation. It required some creative
tinkering and awareness that binary transitive closures are harmful to
make the switch to the solution we presented.  Replacing it with the
unary \texttt{reachable} relation had a major impact on efficiency.
Also our first solution to the stemmatology problem had a transitive
closure. Here, transitive closure could be avoided altogether. It was
a major step forward in efficiency to replace it by a coloring
function for the nodes in the stemma graph. The other two problems
also use functions (\texttt{label} and \texttt{colorOf} respectively)
whose range defines membership in an equivalence class.

The preference of a unary transitive relation over a binary one is an
illustration of another general principle: less variables is better in
rules and constraints. One should try to break up complex rules and
constraints in simpler ones requiring less variables and explore
whether one can do with predicates and functions having less
arguments. Another important point is that one should not be satisfied
with a first correct model. Often, major improvements are possible, as
we illustrated in several of our applications.

The \idpdrie system is an evolving research system and further
improvements are on the way. A lot of ongoing work aims at making the
performance less dependent on clever modeling. One recent feature is
symmetry breaking. Predicate level symmetry detection and dynamic
symmetry breaking (during search) automatically exploit symmetries
present in the problem \cite{ictai/DevriendtBMDD12} (symmetry is
present in the DFA problem: permuting the colors gives another
solution; however it was broken in an ad-hoc way in the SAT encoding
and hence also in the input structure of our instances). One recent
feature is to avoid complete proposionalization during grounding. On
one hand by keeping function terms in the grounding
\cite{ictai/DeCaTB13}, on the other hand through lazy, demand driven
grounding during search \cite{iclp/DeCatDS12,corr/CatDSB14}. Another
feature is the detection of functional dependencies and their use to
reduce the arity of predicates \cite{tplp/CatB13}.

\section*{Acknowledgements}
Caroline Mac\'e and Tara Andrews introduced some of the authors to
stemmatology and provided the data sets; Tara also explained the
working of the procedural code.

This work was supported by Research Foundation - Flanders
(FWO-Vlaanderen) and by the Research Council of KU Leuven (GOA/08/008
and GOA 13/010). 


\end{document}